\documentclass[journal]{IEEEtran}
\usepackage{amsmath}
\usepackage{amssymb}
\usepackage{floatfig,epsfig}
\usepackage{verbatim}
\usepackage{pifont}
\usepackage{algorithm}
\usepackage{algorithmic}
\usepackage{alltt}
\usepackage{bm}
\usepackage{isolatin1}
\usepackage{graphics}
\usepackage{epic}
\usepackage{eepic}
\usepackage{graphpap}
\usepackage{psfrag}
\usepackage{float}
\usepackage{subfigure}
\newtheorem{thm}{Theorem}
\newtheorem{col}{Corollary}
\newtheorem{pro}{Proposition}

\newtheorem{clm}{Claim}
\newtheorem{lem}{Lemma}
\newtheorem{rem}{Remark}
\newtheorem{exmpl}{Example}
\ifCLASSINFOpdf
\else
\fi

\begin{document}
% paper title
\title{Optimized-Cost Repair in Multi-hop Distributed Storage Systems with Network Coding\\
}
\author{\dag Majid Gerami, \dag Ming Xiao, \dag Mikael Skoglund, \ddag Kenneth W. Shum,  *Dengsheng Lin \\ E-mail: \{gerami, mingx, skoglund\}@kth.se
\\ \dag Communication Theory, School of Electrical Engineering, KTH, Stockholm, Sweden \\
\ddag Institute of Network Coding, the Chinese University of Hong Kong, Shatin, Hong Kong,
wkshum@inc.cuhk.edu.hk\\
*National Key Lab of Communications, University of Electronics Science and Technology of China\\ linds@uestc.edu.cn}
\date{\today}

\maketitle

\begin{abstract} In distributed storage systems reliability is achieved through
redundancy stored at different nodes in the network. Then a data
collector can reconstruct source information even though some
nodes fail. To maintain reliability, an autonomous and efficient
protocol should be used to repair the failed node. The repair
process causes traffic and consequently transmission cost in the
network. Recent results found the optimal traffic-storage
tradeoff, and proposed regenerating codes to achieve the
optimality. We aim at minimizing the transmission cost in the repair
process. We consider the network topology in the repair, and accordingly modify  information flow graphs. Then we analyze the cut
requirement and based on the results, we formulate the
minimum-cost as a linear programming problem for linear costs.
We show that the solution of the linear problem establishes a fundamental lower
bound of the repair-cost. We also show that this bound is achievable for
minimum storage regenerating, which uses the optimal-cost
minimum-storage regenerating (OCMSR) code. We propose surviving node cooperation  which can efficiently reduce the
repair cost. Further, the field size for the
construction of OCMSR   codes is discussed. We show the gain of optimal-cost repair in tandem, star, grid and fully connected networks.

\end{abstract}


\section{Introduction}\label{sec:intro}

\IEEEPARstart{T}{he} increasing number of new applications such as voice and video over Internet, video on demands, user generated content (like YouTube service) have
caused a fast growth of traffic over the wired and wireless networks during
the last decades. This traffic growth requires  big data storages. Consequently, the idea of distributing  files among some storage nodes on the network has been proposed. High demand in QoS (quality of service) in e.g.,
fast and ubiquitous access, low-delay  and reliability poses high
requirements on distributed information storage. Compared to centralized storage systems, the benefits of
distributed storage systems include fast and ubiquitous access,
high reliability,  availability and scalability. For the benefits, recently distributed storage systems have been used in big data centers like Google file systems \cite{{Ghem01}}, Oceanstore \cite{Kub01} and Total Recall \cite{Bhagwan01}.

On another aspect, the distributed storage system has brought the higher reliability of the measured data in wireless sensor networks \cite{Dimk03}. Wireless sensor networks consist of several small devices
(sensors) which measure or detect a physical quantity e.g., temperature, pressure, light and so on. The main
characteristic of the sensors are limited battery, low
CPU power and limited communication capability \cite{Jain01}. Because of the limited communication capability of sensors,  communication between nodes is generally through multi-hop. That is, storage nodes may relay the  message of other nodes.
\begin{figure*}%[h!]
 \centering
 \psfrag{n11}[][][1.5]{ $a_1$ }
 \psfrag{n12}[][][1.5]{$b_1$ }
 \psfrag{n21}[][][1.5]{ $a_2$ }
 \psfrag{n22}[][][1.5]{$b_2$ }
 \psfrag{n31}[][][1.5]{ $a_1+b_1+a_2+b_2$ }
 \psfrag{n32}[][][1.5]{$a_1+2b_1+a_2+2b_2$ }
 \psfrag{n41}[][][1.5]{ $a_1+2b_1+3a_2+b_2$ }
 \psfrag{n42}[][][1.5]{$3a_1+2b_1+2a_2+3b_2$ }
 \psfrag{n51}[][][1.5]{ $5a_1+7b_1+8a_2+7b_2$ }
 \psfrag{n52}[][][1.5]{$6a_1+9b_1+6a_2+6b_2$ }
 \psfrag{m1}[][][1.5]{ node 1 }
 \psfrag{m2}[][][1.5]{ node 2 }
 \psfrag{m3}[][][1.5]{ node 3 }
 \psfrag{m4}[][][1.5]{ node 4 }
 \psfrag{m5}[][][1.5]{ node 5 }
\psfrag{t1}[][][1.5]{ $p_1=a_1+2b_1$}
\psfrag{t2}[][][1.5]{$p_2=2a_2+b_2$}
 \psfrag{t3}[][][1.5]{$p_3=4a_1+5b_1+4a_2+5b_2$}
 \psfrag{p1}[][][1.5]{ $p_1$}
 \psfrag{p2}[][][1.5]{ $p_2$}
 \psfrag{p3}[][][1.5]{ $p_3$}
\psfrag{mu1}[][][1.5]{ $\times 1$}
 \psfrag{mu2}[][][1.5]{ $\times 2$}
 \psfrag{mu3}[][][1.5]{ $\times 3$}
  \resizebox{12cm}{!}{\epsfbox{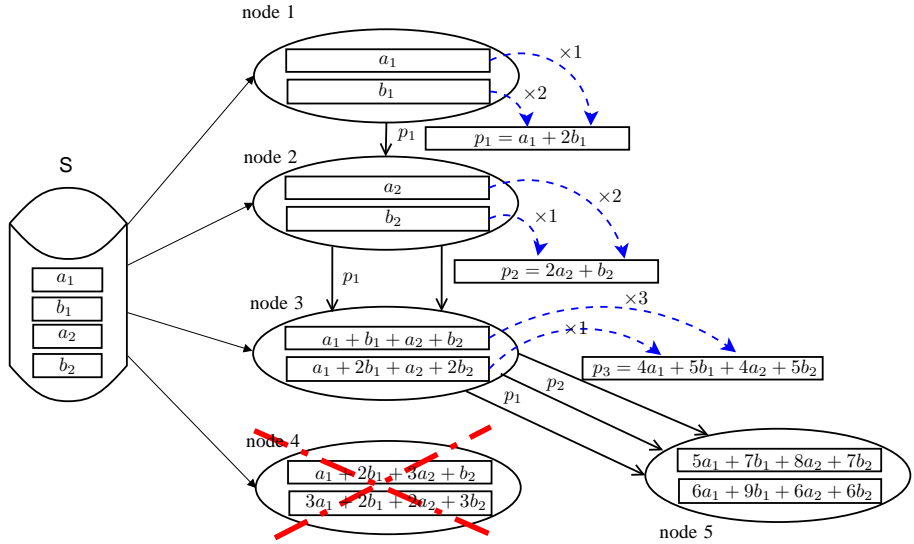}}
 \caption{A distributed storage system in a 4-node tandem network. Dashed lines denote how network codewords are formed. Node $4$ fails and node $5$
  is the new node. For regenerating a new node, $p_1, p_2, p_3$ are formed by linear combination of fragments in node $1$, $2$, and $3$, respectively. The underlying finite field is GF(11).}
  \label{Tandem-Flow}
\end{figure*}

 The storage nodes are vulnerable due to disk failure, power off, or a node leaving the system. Thus to make the
data reliable over the unreliable nodes, the data is encoded and
distributed among  storage devices  \cite{Dimk01}, \cite{Wu01}, \cite{Jain01}. Reliability is usually obtained through redundant nodes in a
distributed storage system, in which error control (EC) codes are
normally used to increase the storage efficiency. An  EC code with the
\textit{MDS (maximum distance separable)} property is optimal in term of
the redundancy-storage tradeoff. If a  file of size $M$
bits is coded by $(n,k)$-MDS codes ($k \leq n$) and distributed
among $n$ nodes, each node  stores $M/k$ bits. Then every $k$ nodes can reconstruct the original file. Nonetheless, if a node fails, and a new node joins the system, the new node needs to download $M$ bits to regenerate $M/k$ bits.  Thus, it may not be optimal considering the traffic for regenerating
a new node. In \cite{Dimk01}, \cite{Wu01}, the authors investigate the problem of
the repairing traffic (bandwidth) and find the optimal
storage-bandwidth tradeoff. A new class of erasure codes, namely
regenerating codes based on network coding
(\cite{InfoFlow}, \cite{KoMed}), are proposed in \cite{Dimk01},
\cite{Wu01} to achieve the optimal tradeoff.  In the repair process,
 the new node may not have the same encoded data as the failed node.  However,
 the restored data has the same MDS property (i. e., with the new node still every $k$ nodes can reconstruct the original file). We call this property as \textit{regenerating code property (RCP)}. This kind of repair
 is generally called functional repair, in contrast to the exact repair,
 where the new node stores the same data as the failed node. The exact regenerating of a new node is studied in \cite{Rashmi}, \cite{Shah02}, \cite{Pap01}.%Throughout this paper,  we first consider the functional repair then we suggest an explicit construction of exact repair.

  Two extreme scenarios in the fundamental storage-bandwidth tradeoff of regenerating codes \cite{Dimk01} correspond to two
  types of regenerating codes: one scenario, namely the minimum storage regenerating (MSR) code, uses
  the same amount storage space as the MDS code but with lower repair traffic.
  Another scenario, namely the minimum bandwidth regenerating (MBR) code has slightly more storage per node, but its repair
  bandwidth is considerably lower than the MSR (MDS) codes \cite{Dimk01}.  To further decrease the repair traffic in the scenario of multiple
simultaneously-failed nodes, reference \cite{Yuchong01} proposes new node cooperative
repair. In new node cooperative regenerating codes, the new nodes cooperate to minimize the repair bandwidth. The difference between new node cooperative regenerating and our method is that in our approach surviving nodes cooperate, i.e., surviving nodes  can combine the received data with their own information. These two approaches together can provide a fully cooperative regenerating process.

Though references \cite{Dimk01}, \cite{Wu01}, \cite{Yuchong01},
\cite{Rashmi} have well addressed the optimal storage-bandwidth
tradeoff in distributed systems from different aspects, the link
cost (transmission cost in channels) and the impact of the network
topology have not been considered. In a practical system, the cost
is an important design consideration and different links
(channels) of a network may have different costs. The topology of
networks may also impact the repair cost. Looking through two main applications for distributed storage systems in data centers and wireless sensor network, we find that data centers have hierarchical network structures \cite{Benson} and wireless sensor networks have multi-hop network structures \cite{Kong01}. Recently, reference
\cite{akh01} considers the transmission cost from surviving
nodes to the new node and the cost-bandwidth tradeoff is
derived. This model  suits for the scenarios in which there are direct links from surviving nodes to the new node. Yet, this model does not exploit the network structure. For instance, consider two storage nodes that have direct links to the new node as well as a link between themselves. Suppose these surviving nodes want to send data to the new nodes. Instead of  direct communication, the surviving nodes could cooperate and send an aggregated data to the new node. Thus, we propose the idea of surviving nodes cooperation in order to reduce the transmission cost. This problem is also interesting in distributed storage systems in wireless sensor networks (WSN), where reducing transmission cost is demanding. In  \cite{Maj02} it has been shown that the optimal-cost repair can be  found also in a decentralized approach which suits for WSN due to the lack of a central node and CPU power limitation of nodes.

Other related work includes followss. Reference \cite{Li01}
proposes a tree-structure algorithm for minimum storage regenerating (MSR) codes. The work considers a network having links with different bandwidth (capacity). Then by exploiting Prim's algorithm, they  find a repair approach which minimizes the repair time (i.e., finding routes in network with the highest capacity to maximize the speed of repair process). Our work is distinct from \cite{Li01} since we  minimize the repair-cost instead of repair-time. Repair-cost can be formulated as the summation of costs on all links in network, while repair-time depends on the bottleneck of a route.
%However, authors prove Their algorithm is optimum based on wrong assumption. We shall show by a contradict example that the proposed algorithm does not hold in general network. \cite{Shah01} suggests a flexible regenerating codes. The difference with our work is that they aim to minimize traffic in a specific network that all surviving nodes has direct connection to the new nodes. We aim at minimizing repair cost which in general might be different with minimum traffic repair. Besides, our approach holds for every network topology.

%In this paper, the repair-cost is summation of the link costs in the repair process. For some other cost like delay it does not hold. We restrict our attention to the linear cost, though the results  can be readily extended to more general scenarios (e.g., convex costs).We suggest a method that for MSR code, cut analysis in any repair stage is independent of future stage and just cut analysis in that stage is enough to find the optimum-cost repair. this is not true for other points in tradeoff

 We shall study the repair-cost in multi-hop networks where links may have  different transmission costs.  To formulate the problem, we modify the information flow graph by taking into considering the network topology. Then, by cut-set bound analysis and solving a linear programming problem, we derive  a lower bound of repair
costs. In distributed storage systems, networks evolve along with time. Thus, there is infinite stages of repair. In general, one should run cut-set bound analysis in infinite stages to find the optimum cost.  We first give a lower bound of repair-cost by cut-set analysis in only first stage of repair.  Later, we prove that this lower bound is tight for minimum storage regenerating codes which are the optimum-cost minimum storage regenerating codes. To the best of our knowledge, we are the first who consider \textit{surviving node
cooperation (SNC)} scheme to reduce the repair cost. SNC allows intermediate nodes in the network to combine their received fragments from other surviving nodes and their own stored data.
For a general setting, we formulate the minimum-cost problem
and study the solution. Related work includes
\cite{Lun01}, in which the minimum-cost multicast problem is
studied in networks with network coding. Reference \cite{Jiang01} studies the
transmission cost of distributing a source file among nodes on the
network with or without network coding.
It is shown in \cite{Jiang01} that by file splitting the optimality can be
achieved. Here we assume a source file is already
distributed among different nodes. Our objective is to find the optimum
repair process when a storage node fails, i.e., how to regenerate a new node with the minimal cost. It may worth to note that our approach holds for every network topology. The algorithm takes the cost of transmitting one unit of information between each pair of nodes as its input and find the optimum cost repair. In addition, we try to find a closed form for the repair cost for some specific topologies like tandem, star, grid, and fully connected networks.

The remainder of the paper is organized as follows. In Section
\ref{Sec:sysMod}, we will give a motivation example by a specific
 network. For more general networks, we formulate optimization
problem in Section \ref{sec:PrbFrm}.  In Section \ref{sec:code-construction}, we
discuss the optimal-cost repair when the storage capacity per node is $\alpha= M/k$. In Section
\ref{sec:lowerbound} we investigate the  fundamental lower bound of repair costs in tandem, star, grid, and fully connected networks and present the gain of the optimal-cost repair in those networks.
Finally, in Section \ref{sec:conclusion}, we conclude the
paper.

\begin{table}[ht]
\caption{Notations} % title of Table
\centering % used for centering table
\begin{tabular}{c l } % centered columns (2 columns)
\hline\hline %inserts double horizontal lines
Symbol & Definition \\[0.5ex] % inserts table %heading
\hline % inserts single horizontal line
$n$ & number of storage nodes in a distributed storage\\ & system\\ % inserting body of the table
$k$ & minimum number of nodes needed to- \\ & reconstruct the original file \\
$d$ & number of surviving nodes for a repair process \\
$\alpha$ & storage amount of individual node \\
$\beta$ & amount of data downloads from each \\ & surviving node in the repair process \\
[1ex]
$M$ & original information file size\\
$\underline{s}$ & $M\times 1$ vector denoting the original file \\
$\underline{Q}_{i}$ & $M\times \alpha$ matrix denoting coding  \\ & coefficients of node $i$\\
$\underline{X}_{i}$ & $\alpha \times 1$ vector denoting the content of node $i$,\\ i.e., $\underline{X}_{i}=\underline{Q}_{i}^{T}\underline{s}$ \\
$\sigma_{non-opt}$ & repair cost by non-optimized-cost approach\\ &(repair-cost without SNC)\\
$\sigma_c$ & repair cost by optimized approach\\ &(repair-cost with SNC and optimization)\\
$g_c$ & optimization gain, $g_c=\sigma_{non-opt}/\sigma_c$\\
%MSR & minimum storage regenerating ($\alpha=M/k$)\\
%$\underline{A} & underlined uppercase letter denotes matrix\\
%$\underline{a} & underlined lowercase letter denotes vector\\
%
%\\[1ex] adds vertical space
\hline %inserts single line
\end{tabular} \label{table:notation} % is used to refer this table in the text
\end{table}

\begin{figure}%[b]
 \centering
 \psfrag{a}[][][1.5]{ $\alpha$ }
 \psfrag{b}[][][1.5]{ $\beta$ }
 \psfrag{i}[][][1.5]{ $\infty$ }
 \psfrag{m1}[][][1.5]{ node 1 }
 \psfrag{m2}[][][1.5]{ node 2 }
 \psfrag{m3}[][][1.5]{ node 3 }
 \psfrag{m4}[][][1.5]{ node 4 }
 \psfrag{m5}[][][1.5]{ node 5 }
 \psfrag{x11}[][][1.5]{$x_{in}^1$ }
 \psfrag{x12}[][][1.5]{$x_{out}^1$ }
 \psfrag{x21}[][][1.5]{$x_{in}^2$ }
 \psfrag{x22}[][][1.5]{$x_{out}^2$ }
 \psfrag{x31}[][][1.5]{$x_{in}^3$ }
 \psfrag{x32}[][][1.5]{$x_{out}^3$ }
 \psfrag{x41}[][][1.5]{$x_{in}^4$ }
 \psfrag{x42}[][][1.5]{$x_{out}^4$ }
 \psfrag{x51}[][][1.5]{$x_{in}^5$ }
 \psfrag{x52}[][][1.5]{$x_{out}^5$ }
 \resizebox{8cm}{!}{\epsfbox{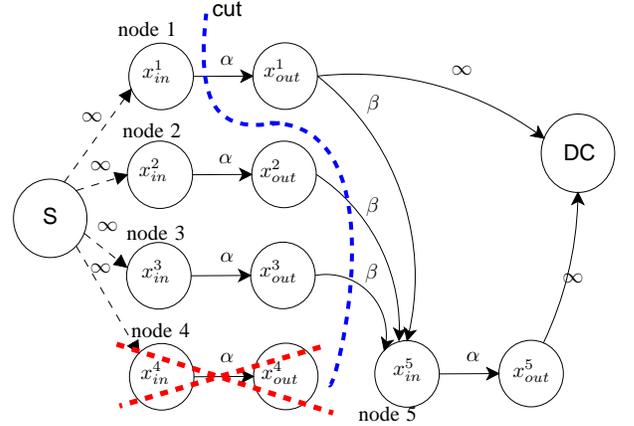}}
 \caption{Cut-set bound analysis in the information flow graph. Node $4$ fails and node $5$ is the new node. There exist direct links from surviving nodes to the new node.}
 \label{cut}
\end{figure}

\section{Motivating Example}\label{Sec:sysMod}

We first give an example to illustrate the motivation. Consider a
distributed storage system in a four-node tandem network shown in
Fig. \ref{Tandem-Flow}. Here we follow the notation of the
regenerating codes \cite{Dimk01}, \cite{Wu01}. The notations are also
given in Table \ref{table:notation}. We divide a source file into equal size fragments (packets). A file of size $4$
mega-bits ($ M =4$ fragments;  One fragment has $1$ mega-bit) is coded with a regenerating code
\cite{Dimk01} and distributed among $4$ nodes ($n=4$). Each node
stores $2$ one-mega-bit fragments ($\alpha = 2$) and the source
file can be reconstructed by any $2$ nodes ($k = 2$). When a node
fails (say node $4$), a new node downloads $\beta$ fragments from
each of $3$ surviving nodes ($d = 3$).

If the new node has direct links to all the surviving nodes,  the distributed storage system can be
represented by a directed acyclic graph with a failure/repair on node $4$, as shown in Fig. \ref{cut}. The graph is known as an information flow graph \cite{Dimk01}. In the graph, there is a source node
($S$) connected to storage nodes through
infinite-capacity links. Each storage node is depicted by input
($in$) and output ($out$) nodes connecting by a $\alpha$-capacity link. When a node fails, a new node downloads
$\beta$ fragments from each node. A $DC$ denoting a data collector
reconstructs the source file from any $k$ storage nodes.
\begin{figure*}%[b]
 \centering
 \psfrag{n11}[][][1.5]{ $a_1$ }
 \psfrag{n12}[][][1.5]{$b_1$ }
 \psfrag{n21}[][][1.5]{ $a_2$ }
 \psfrag{n22}[][][1.5]{$b_2$ }
 \psfrag{n31}[][][1.5]{ $a_1+b_1+a_2+b_2$ }
 \psfrag{n32}[][][1.5]{$a_1+2b_1+a_2+2b_2$ }
 \psfrag{n41}[][][1.5]{ $a_1+2b_1+3a_2+b_2$ }
 \psfrag{n42}[][][1.5]{$3a_1+2b_1+2a_2+3b_2$ }
 \psfrag{n51}[][][1.5]{ $2a_1+3b_1+3a_2+2b_2$ }
 \psfrag{n52}[][][1.5]{$3a_1+6b_1+3a_2+3b_2$ }
 \psfrag{m1}[][][1.5]{ node 1 }
 \psfrag{m2}[][][1.5]{ node 2 }
 \psfrag{m3}[][][1.5]{ node 3 }
 \psfrag{m4}[][][1.5]{ node 4 }
 \psfrag{m5}[][][1.5]{ node 5 }
\psfrag{t1}[][][1.5]{ $p_1=a_1+2b_1$}
\psfrag{t2}[][][1.5]{$p_2=a_1+2b_1+2a_2+b_2$}
 \psfrag{t3}[][][1.5]{$p_3=2a_1+4b_1+a_2+b_2$}
\psfrag{t2}[][][1.5]{$p_2=a_1+2b_1+2a_2+b_2$}
 \psfrag{t3}[][][1.5]{$p_3=2a_1+4b_1+a_2+b_2$}
 \psfrag{t4}[][][1.5]{$p_4=2a_1+3b_1+3a_2+2b_2$}
 \psfrag{t5}[][][1.5]{$p_5=3a_1+6b_1+3a_2+3b_2$}
 \psfrag{p1}[][][1.5]{ $p_1$}
 \psfrag{p2}[][][1.5]{ $p_2$}
 \psfrag{p3}[][][1.5]{ $p_3$}
\psfrag{p4}[][][1.5]{ $p_4$}
 \psfrag{p5}[][][1.5]{ $p_5$}
 \psfrag{mu1}[][][1.5]{ $\times 1$}
 \psfrag{mu2}[][][1.5]{ $\times 2$}
 \psfrag{mu3}[][][1.5]{ $\times 1$}
 %\psfrag{q}[][][.9]{Quantizer}
% \psfrag{q1}[][][.8]{Quantization}
% \psfrag{e}[][][.8]{Error}
% \psfrag{s}[][][1]{Shaping }
% \psfrag{a}[][][.8]{$-$}
% \psfrag{b}[][][.8]{$+$}
% \psfrag{v}[][][.9]{Vector}
%\centerline{\includegraphics[height=100mm,width=80mm]{tandemFW3.eps}}
 \resizebox{12cm}{!}{\epsfbox{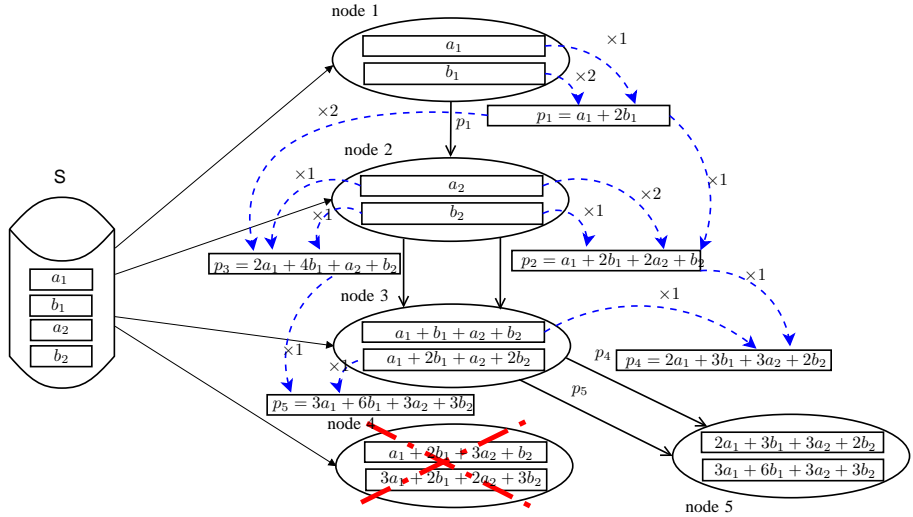}}
 \caption{Regenerating by surviving node cooperation in a tandem network. Here, for regenerating a new node, fragment $p_1$ is formed by combining fragments of node $1$. Fragments $p_2, p_3$ are formed by linear combination of fragment $p_1$ and stored fragments in node $2$. Finally, fragments $p_3, p_4$ are formed by linear combination of $p_2,p_3$, received fragments on node $3$, with the stored fragments on node $3$.}
  \label{Tandem-NC}
\end{figure*}

By cut-set bound analysis on the information flow graph, it is shown that there is the fundamental storage-bandwidth tradeoff to regenerate a new node \cite{Dimk01}. A cut
in the flow graph refers to a set of edges in which  network nodes are
divided into two separated parts; one part contains $S$  and
another part contains the DC. The value of a cut is the sum of
capacities of edges from the source to the destination direction. Then, the corresponding cut must meet a fundamental requirement resulting from min-cut max-flow theorem for multicast networks  \cite{InfoFlow}. The theorem expresses that in a multicast network  destinations can  reconstruct the source
file whenever all the cuts between source and destinations are greater than
or equal to the source file size.

By the cut-set bound analysis of the information flow  graph in Fig. \ref{cut},
we can see that for a DC to reconstruct the source
file, it requires $\alpha+2\beta\geq M$. For $\alpha=2$,  $M=4$,
then $\beta \geq 1$. That is, the new node must download at least
$\beta=1$ mega-bits from each surviving node. The optimal-traffic
repair to achieve this lower bound is to download $p_1, p_2, p_3$
as in Fig. \ref{Tandem-Flow} from node $1$, node $2$ and node $3$,
respectively. Here $p_1$, $p_2$ and $p_3$ are formed by linear
coding at nodes 1, 2, and 3, respectively. Yet, the analysis is
different if we consider the link cost.

We assume that each fragment of $1$ mega-bits transmitted in a
single channel costs one transmission unit. To reach the new node
(node $5$), we can easily see that $p_1$ passes the route
(node $1 \mapsto$ node $2 \mapsto$ node $3
\mapsto$ node $5$) with a cost of $3$ units, and $p_2$ passes node $2\mapsto$ node 3 $\mapsto $ node $5$
with a cost of $2$ units, and $p_3$ passes node $3 \mapsto $ node $5$
with a cost of $1$ unit. Thus, the total cost in the repair is $6$
units. In this paper, we call \textit{non-optimized-cost repair} to the approach which minimizes the repair-bandwidth (without optimizing the cost) and  the corresponding cost  is denoted as  $\sigma_{non-opt}$).

Non-optimized-cost repair may be optimal in term of repair bandwidth. Yet it might not be optimal in terms of the repair cost. For example, consider the repair scheme in Fig. \ref{Tandem-NC}, where we allow surviving node cooperation (SNC). Here the cooperation
means that one surviving node can combine/encode the coded symbols of
another node. That is, surviving nodes are allowed to linearly
combine their own fragments with the received fragments from other nodes. For
instance, at node $2$, $p_2$ and $p_3$ are encoded from the
received $p_1$ from node $1$ and stored fragments in node $2$. At node
$3$, $p_4$ and $p_5$ are obtained by encoding $p_2$ and $p_3$ with
the fragments of node $3$. Note that here we only consider functional
repair, in which the regenerated node may not be identical to
the failed node but it has the same code property. That is, with
the new node (node $5$), any $2$ out of $4$ nodes can reconstruct the source.
It is easy to see that the cost of repair is reduced to $5$ units
as shown in Fig. \ref{Tandem-NC} (only two fragments are transmitted
from node $3$ to node $5$). The example shows SNC can reduce the transmission cost. We note that the repair approach in Fig. \ref{Tandem-NC}  still is not optimal in term of the transmission cost. We shall show in the next section that the optimal repair cost in this example is $4$. We further note that SNC can be applied
to the scenario of one or multiple node failure, but the
 cooperative regenerating codes in \cite{Yuchong01} and \cite{Ken01} can only be used for the  scenario of
multiple node failure.

\section{Problem Formulation}\label{sec:PrbFrm}
In this section, we first introduce the cost matrix which presents the cost of transmitting one unit of data (e.g., one fragment) between any pair of helper nodes. In our approach surviving nodes cooperate in the repair process. Thus, we assume storage nodes are capable of performing  linear calculation in  finite fields. The number of helper nodes, which is denoted as $d$, is assumed to be greater than $k$, i.e., $d\geq k$ (for $d < k$ successful regenerating is not possible; see details in \cite{Dimk01}). Furthermore, we consider the network structure in the information flow graph and  formulate the optimal-cost problem. Our main tool for analysis would be cut-set bound analysis in multicast networks, proposed in seminal paper \cite{Raymond}. In the literature, the process
that a node fails and a new node is regenerated is called a stage
of repair. Evolving network in a distributed storage system (infinite time leaveing/joining of nodes) implies infinite number of repair stages. Cut-set bound analysis with considering  infinite stages of repair in a heterogenous network even for a distributed storage system with small number of storage nodes would be complicated. We thus try to find a lower bound of repair-cost by cut-set bound analysis in the first stages of repair. We later show that this bound can be achieved for the minimum storage regenerating codes, i.e., $\alpha=M/k$. In the literature, the process
that a node fails and a new node is regenerated is called a stage
of repair.

\subsection{Network Setup}
Above,  for a specific network, we have studied the repair cost and proposed SNC to
reduce the cost. Naturally, we may ask what
is the optimal-cost and how to achieve the optimality for more
general scenarios. In this section, we first formulate  a linear
optimization problem which establishes a fundamental lower bound on the repair cost for general networks.

Assume in a distributed storage systems with $n$ nodes, node $n$ fails and $d$ number of surviving nodes ($k\leq d \leq n-1$) help to regenerate a new node. We also assume that network topology and the cost of transmitting one unit of data between nodes are given. From the network topology the pathes from helper nodes to new node are known. To study the optimal-cost repair problem for
a given network, assuming  failure on node $n$, we  define an $d\times n$ cost matrix $\underline{C}$, as follows,
\[ \mbox{\emph{\underline{C}= }} \left(
\begin{array}{cccc}
$0$ & c_{(12)} &\ldots & c_{(1n)} \\
c_{(21)} & $0$ &\ldots & c_{(2n)} \\
\vdots & \vdots & \ddots & \vdots \\
c_{(d1)} &\ldots &\ldots & c_{(dn)}
 \end{array} \right),\]
where an element $c_{(ij)}$ denotes the link cost of transmitting one unit fragment from node $i$ to
node $j$. For instance, $c_{(12)}$ is the cost of transmitting one fragment from node $1$ to node $2$, and $c_{(in)}$ is the corresponding cost from surviving node $i$ to the new node (assuming node $n$ fails). $c_{(ij)}$ is nonzero  if there is  direct link
from node $i$ to node $j$. Otherwise, $c_{(ij)} =\infty$. We assume that an algorithm extracts the matrix $\underline{C}$ from a given network topology and link costs. Clearly, by matrix $\underline{C}$, we can calculate the costs of a path (maybe multi-hop) between any pair of nodes in the network.
Here, we assume there is at least one path between any pair of nodes. Furthermore, we assume that the network is delay-free and acyclic. In this paper, we only consider
linear costs. That means if the transmission cost of one fragment from node
$i$ to $j$ is $c_{(ij)}$, then it costs $m c_{(ij)}$ to
transmit $m$ fragments from node
$i$ to $j$.

To investigate the repair-cost in a given network, we modify
the information flow graph in \cite{Dimk01} by introducing the network
topology and link costs into the graph.

\subsection{Modified Information Flow Graph }

Consider a storage system with the original file of size $M$
distributed among $n$ nodes in which each node stores $\alpha$
fragments and any $k$ of $n$ nodes can rebuild the original
file. We denote the source file with an $M\times1$ vector,
$\underline{s}$. Then, the code on  each node $i$ can be represented by a matrix
 $\underline{Q}_i=(\underline{q}_i^1,\cdots,\underline{q}_i^{\alpha})$ of size $M\times\alpha$ where each
 column ($\underline{q}_i^j$) represents the code coefficients of  fragment $j$
 on node $i$, and then the stored data in node $i$ is $\underline{X}_i=\underline{Q}_i^{T}\underline{s}$.

We can denote the flow of
information in a distributed storage system by a directed acyclic
graph  $G(n,k,\alpha)=G(N,A)$, where $N$ is the set of
nodes and $A$ is the set of directed links.
Similar to \cite{Dimk01}, graph $G(n,k,\alpha)$ consists of three
different nodes: a source node, storage nodes and several data collectors
($DC$). The source node distributes the original file among
storage nodes along with the  (assumably) infinite-capacity links.
Every storage node can be denoted by input ($in$) and output ($out$) nodes
connecting by a link of capacity equals to the amount of node's
storage ($\alpha$). Finally, the $DC$ reconstructs the original file by connecting to at least
$k$ storage nodes via the infinite-capacity links and then solving  linear equations. Contrary to \cite{Dimk01}, there might not exist direct links from surviving nodes to the new node and storage nodes may relay other nodes' fragments to the new node. When a node fails, $d$ surviving nodes participate in the repair process ($k \leq d \leq n-1$). Here,  an optimization algorithm shall determine the optimum traffic on the links of the network. We note that in distributed storage networks, node failure  and new node generation may happen infinite times. In a heterogenous network, considering all stages of repair makes the problem complicated. Hence, in our study we focus on the first stage of repair, which is a performance lower bound for infinite repair stages. As an example, the modified information flow
graph for the first stage of repair on node $4$ has been shown in Fig.~\ref{modifiedgraph}. In this example, nodes are connected in a tandem network.
\begin{figure}%[b]
 \centering
 \psfrag{s}[][][1.5]{ S }
 \psfrag{a}[][][1.5]{ $\alpha$ }
 \psfrag{i}[][][1.5]{ $\infty$ }
 \psfrag{m1}[][][1.5]{ node 1 }
 \psfrag{m2}[][][1.5]{ node 2 }
 \psfrag{m3}[][][1.5]{ node 3 }
 \psfrag{m4}[][][1.5]{ node 4 }
 \psfrag{m5}[][][1.5]{ node 5 }
 \psfrag{z12}[][][1.5]{ $z_{(12)}$ }
 \psfrag{z23}[][][1.5]{ $z_{(23)}$ }
  \psfrag{z35}[][][1.5]{ $z_{(35)}$ }
 \resizebox{8cm}{!}{\epsfbox{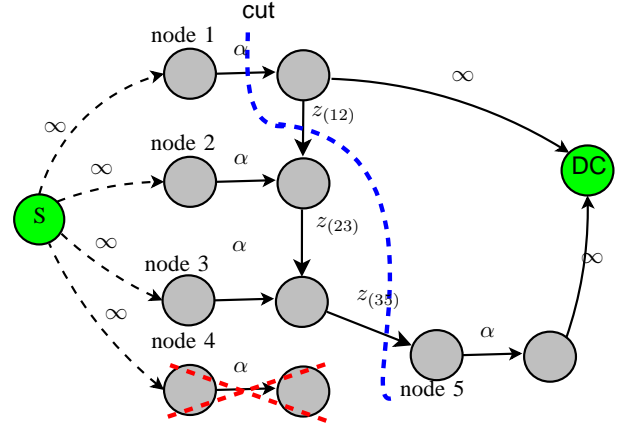}}
\caption{Modified information flow graph for the first stage of repair in a  tandem network. The cut mentioned on the figure with heavy dotted-line corresponds to the inequality $  z_{(35)}+\alpha \geq M$.}
 \label{modifiedgraph}
\end{figure}

We use a column vector to denote the number of fragments
transmitted on the directed links. This vector is termed as a
\emph{subgraph} ($\underline{z}=[z_{(ij)}]_{\mid_{(ij)\in A}}$)
\cite{Lun01}, where each element $z_{(ij)}$ is a non-negative integer and represents the number of fragments from node $i$ to node $j$ for $(ij)\in A$. For the subgraph
$\underline{z}=[z_{(ij)}]_{\mid_{(ij)\in A}}$, and defining $c_{(ij)}$ as the cost of transmitting one unit of data (e.g., one fragment) from node $i$ to node $j$, the repair cost can be formulated as,
\begin{equation}
\sigma_c(\underline{z}) \triangleq \sum_{(ij)\in A} c_{(ij)}z_{(ij)}.
\end{equation}
Our objective
is to minimize the cost ($\sigma_c)$ considering only the first stages of repair.

\rem{We note that in the repair problem there would be infinite number of repair stages. However,   cut analysis in the first stage of repair gives the necessary  conditions for finding the optimum-cost regenerating approach. The analysis  gives a lower bound of repair-cost for the repair. Later, we show that this lower bound can be achieved for $\alpha=M/k$. Thus, in our analysis there would be limited number of cut-set constraints. Clearly, the number of cut constraints  depends on the number of nodes $n$ and $k$.}

\subsection{Cut-set analysis in the first stage of repair}
\subsubsection{Constraint Region} \label{constregion}
In the repair process, it is required that any $k$ nodes can
reconstruct the original file. This property is known as the
regenerating code property (RCP).  The RCP must be preserved for an arbitrary number of stages of repair.
 For this, it is necessary for the first stage of repair all the cut constraints to be satisfied.  Thus, we find the minimum $\sigma_c$ under the constraint that in
 the first stage of repair, all cuts of connecting the $DC$ to the new node
 and $k-1$ other nodes must be greater than or equal to $M$. In Fig. \ref{modifiedgraph}, the heavy dotted line presents a cut when the DC connects to the new node and node $1$ for a regenerating code with parameter $k=2$. The cut constraint relating to this cut can be formulated by the inequality: $ z_{(35)}+\alpha \geq M$. By assuming vector $\underline{z}=[z_{(12)} z_{(23)} z_{(35)}]^T$, we can express the inequality in vector space as, $(0, 0 , 1) \underline{z}\geq M-\alpha$.   Assume  $r$  cut constraints on the first stage of repair. Denoting $|A|$ the cardinality of existing edges between nodes, we form all the inequalities in a matrix form, by defining an $r \times |A|$ dimension matrix $L$ (this matrix  is called \emph{coefficient matrix} in the literature \cite{Optbook}). The corresponding inequalities induced by the cut constraints show a region in a multi-dimensional space that the subgraph must satisfy to be a feasible solution. This region is often called \textit{polytope} \cite{Optbook}. Consequently
the polytope is
\begin{equation}
\Psi=\{ \underline{z}=[z_{(ij)}] \mid  z_{(ij)} \geq 0, \underline{L} \text{ } \underline{z}\geq \underline{b}
\}, \label{eqn:constraint}
\end{equation}
where the comparison of two vectors e.g, $\underline{a} \geq
\underline{b}$ means every element in $\underline{a}$ is greater
than or equals to the element in $\underline{b}$ at the same
position.

\begin{exmpl} In Fig. \ref{modifiedgraph}, if the DC connects to node $1$
and the new node to rebuild the source ($k = 2$),
the first  cut constraint is
\begin{equation}
z_{(35)} \geq M-\alpha.
 \end{equation}
The second constraint follows if we connect the DC with  node $2$
and the new node,

 \begin{equation}
z_{(12)}+ z_{(35)} \geq M-\alpha.
 \end{equation}
Similarly, when the DC connects to node 3 and the new node, we have
the third constraint
\begin{equation}
z_{(23)} \geq M-\alpha.
 \end{equation}
\end{exmpl}
Thus, we can form all these inequalities in a matrix form as follows,
\begin{equation}\underbrace{
\left( \begin{array}{ccc}
0 &0 &1\\
1 &0 &1\\
0 &1 &0
 \end{array}\right)}_{\underline{L}}  \underbrace{\left[ \begin{array}{l} z_{(12)}\\ z_{(23)}\\ z_{(35)} \end{array} \right]}_{\underline{z}} \geq \underbrace{ \left[ \begin{array}{l} M-\alpha\\ M-\alpha\\ M-\alpha \end{array}\right]}_{\underline{b}}. \label{Lexample}\end{equation}.
\rem We note that in the cut-set bound analysis there would be some cut constraints that do not affect the polytope region. These are called non-active constraints. In the example above we only consider the active constraints.

\begin{rem}
The polytope $\Psi$ is restricted by linear inequalities. Hence, if
$z_{(ij)}$s are real numbers then the constraint region $\Psi$ is
convex. We can reasonably assume that $z_{(ij)}$s are real
numbers. Note that the file is measured by bits (integer) but it is normally
quite large. Hence we can consider $z_{(ij)}$s as real values since one
fragment has lots of bits. Following this assumption, $\Psi$
constitutes a convex region.\end{rem}

\subsubsection{Linear Optimization} \label{sec:linearOpt}

Since the constraint region and objective function in the repair problem are linear,   the problem is a linear
optimization problem which can be solved efficiently. Finally,  we can
formulate the optimization problem as,
\begin{equation}
 \begin{array}{lc}
 \mbox{\text{minimize}} &  \sigma_c(\underline{z})=\sum_{(ij)\in A} c_{(ij)} z_{(ij)}  \\
 \mbox{\text{subject to}} & \underline{L} \text{ } \underline{z}\geq \underline{b},  \\
 & z_{(ij)}\geq $0$.
 \end{array}
 \label{opt-lin}
 \end{equation}

The linear programming in  (\ref{opt-lin}) results in a lower bound of repair costs.

\begin{pro} The repair-cost calculated by
 problem (\ref{opt-lin}) is the lower bound of the repair-cost.\end{pro}

\begin{proof} In a general network, cut-set bound analysis considering infinite stages of repair is needed to  have all the constraint on $z_{(ij)}$s. This is the necessary and sufficient conditions on $z_{(ij)}$ for the repair. Thus, cut-set bounds considering only the first stage of repair gives only necessary conditions, and  in general they might not be sufficient conditions. Denote the polytope resulting from cut-set analysis on the first stage as $\Psi^1$, and  the polytope given the cut-set constraints till stage $t$ as $\Psi^t$, for $t$ as a positive integer number. We have $\Psi^t \subseteq \Psi^1$, because for any $\underline{z}  \in \Psi^t$, satisfying constraint till stage $t$, $\underline{z}$ have to satisfy the cut constraints at first stage, i.e., $\underline{z}  \in \Psi^1$. Denote $\sigma_{opt}^{t}$ as the minimum value of $\sigma_c$ satisfying all the cut constraints till stage $t$. In other words, $\sigma_{opt}^{t}$  corresponds to the minimum $\sigma_c$ on which plane  $\sigma_c=\sum_{(ij)\in A} c_{(ij)}z_{(ij)} $ intersects polytope $\Psi^t$.  Since $c_{(ij)}$s are non-negative real numbers, $\Psi^t \subseteq \Psi^1$, the minimum-cost given the cut constraints till stage $1$ is not greater than the minimum-cost given the cut constraints till stage $t$. That is $\sigma_{opt}^{1} \leqslant \sigma_{opt}^{t}.$
\end{proof}
%traffic on link $z_{(ij)}^1, (ij)\in A$ are amount of traffic in first stage oFirst stage constraints give a polytope of $z_{(ij)}, (ij)\in A$. Considering next stages of repair gives other constraints. These constraints to be active, they should pose higher traffic on links $z_{(ij)} which result in a higher cost.
%By solving problem \ref{opt-lin}, the minimum-cost subgraph is
%found. Corresponding the minimum cost subgraph for $\alpha=M/k$
%  we can always find a linear coding
%scheme for the repair, satisfying RCP.
%Thus, having constraint on $z_{(ij)}, (ij)\in A$ resulting from cut-set analysis in the first stage,  cut-set analysis considering next stages $z_{(ij)}$s would not be less than the first stage. Since $c_{ij}$ is non-negative real number then $\sigma_c$ would be greater than the repair cost value found in fist stage of repair.
 \subsubsection{Example for the 4 node tandem network}\label{sec:LpExamp1}

Now we have enough tools to find the optimum repair cost in the motivating example (Fig.~\ref{Tandem-Flow}). We assume three nodes joining the repair process ($d=3$), $M=4, k=2, \alpha=2$, and the corresponding cost matrix is,
\begin{equation}
\underline{C}=
\begin{pmatrix}
0 &1 &\infty &\infty\\
\infty &0 &1 &\infty\\
\infty &\infty &0 &1\\
\end{pmatrix}.
\end{equation}
Clearly there exists a direct link between nodes $1$ and $2$ with one unit cost, and similarly
between nodes $2$ and $3$, and nodes $3$ and $5$ (new node).
Now we analyze the constraint region for $M=4, \alpha=2 $ in (\ref{Lexample}).
Hence, we can formulate the problem as
\begin{equation} \begin{array}{lc}
\mbox{\text{minimize}} & \sigma_c(\underline{z})=z_{(12)}+ z_{(23)} +z_{(35)}\\
%\vspace{10 mm}
\mbox{\text{subject to}} &
\left( \begin{array}{ccc}
0 &0 &1\\
1 &0 &1\\
0 &1 &0
 \end{array}\right)\left[ \begin{array}{l} z_{(12)}\\ z_{(23)}\\ z_{(35)} \end{array} \right] \geq \left[ \begin{array}{l} 2\\ 2\\ 2 \end{array}\right]\end{array}.\end{equation}

Solving the linear optimization problem (e.g., by a simplex method
\cite{Optbook}) gives the optimal subgraph
$(z_{(12)},z_{(23)},z_{(35)})=(0,2,2)$ with a cost of 4 units. A linear
network coding by selecting coefficients (in this example) from $GF(5)$  with SNC in Fig. \ref{Tandem-NC} can meet the minimum-cost
subgraph. The coding scheme is  $\underline{p}_1=0$, $\underline{p}_2=2\underline{a}_2+\underline{b}_2$, $\underline{p}_3=\underline{a}_2+2\underline{b}_2$,
$\underline{p}_4=\underline{p}_2+(\underline{a}_1+\underline{b}_1+\underline{a}_2+\underline{b}_2)=\underline{a}_1+\underline{b}_1+3\underline{a}_2+2\underline{b}_2$ and
$\underline{p}_5=\underline{p}_3+(\underline{a}_1+2\underline{b}_1+\underline{a}_2+2\underline{b}_2)=\underline{a}_1+2\underline{b}_1+2\underline{a}_2+4\underline{b}_2$. Here $\underline{p}_4, \underline{p}_5$
are fragments for the new node, which satisfies RCP.

\section{ Optimal-Cost Code  for the Minimum Storage Network ($\alpha=M/k$) }\label{sec:code-construction}
%\subsection{Achievable Lower Bound  for  $\alpha=M/k$ and  the Required Field Size}
In this section, we  show that the lower bound of the repair-cost is achievable for $\alpha=\frac{M}{k}$.  That is, there exists a linear network code corresponding to the repair with the minimum-cost subgraph from the optimization problem (\ref{opt-lin}). We call the codes that achieve this optimal point  as the optimal-cost minimum-storage regenerating (OCMSR) code. Our proof is based on random linear codes and then we discuss the required finite field size for constructing the OCMSR code. Similar to the method in \cite{Wu01}, we  consider  the first stage of repair and then by induction on the number of repair stages, we  generalize the results to multiple stages of repairs. To find the sufficient field size for successful regeneration, we apply sparse-zero lemma as follows.

\begin{lem}\label{schwartz}  Consider a
multi-variable polynomial $g(\alpha_1,\alpha_2,...,\alpha_n)$ which is not identically  zero, and has the maximum degree
  in each variable at most $d_0$. Then, there exist variables  $\gamma_1,\gamma_2,...,\gamma_n$ in the finite field $GF(q)$, and $q \geq d_0$, such that $g(\gamma_1,\gamma_2,...,\gamma_n)\neq 0$
 \end{lem}
 \begin{proof} See proof of Lemma 19.17 in \cite{Raymond}. \end{proof}

Suppose a source information file  consists of $M=k\alpha$ fragments. Assuming each fragment $\underline{m}_i, (i = 1, \cdots, M)$ being a vector of elements in  $\mathbb{F}_q$ %\footnote{A fragment can be a vector of $\mathbb{F}_q$ with the same encoding/decoding operations for all elements.},
we can denote the source by a vector $\underline{s}=[\underline{m}_1,\cdots, \underline{m}_{M}]^T$. Then vector $\underline{s}$ is encoded by an erasure code (satisfying the RCP) to $n\alpha$ fragments and
distributed among $n$ nodes such that every storage node stores $\alpha$ fragments. If $\underline{X}_{i}$ denotes the stored symbols of node $i$, then $\underline{X}_{i}=\underline{Q}_i^{T}\underline{s}$, where $\underline{Q}_i$ is an $M \times \alpha$-dimensional matrix representing the coding coefficients of node $i$. When a node
fails (without loss of generality, we assume node 1 using the encoding coefficients $\underline{Q}_1$ fails) the
optimization algorithm finds the minimum-cost subgraph. Following the minimum-cost subgraph, the new node is regenerated by surviving node cooperation. Clearly, with the minimum-cost subgraph, we also know which nodes should encode on the directed graph. Then, with a proper finite field, we can find the network code and regenerate  the new node with the coding coefficients $\underline{Q}_1^{'}$, and $\underline{X}_{i}^{'}=\underline{Q}_1^{'T}\underline{s}$.

To maintain the RCP after the regeneration of the lost node, the coding coefficients  ($\underline{Q}_1^{'}$) have to meet certain requirement. That is, for any selection of $k-1$  out of $n-1$ surviving nodes, $\Xi_{k-1}=\{\underline{Q}_{s_1},\cdots,\underline{Q}_{s_{k-1}}\}$, together with the codes of the new node $\underline{Q}_1^{'}$, the polynomial $\det([\underline{Q}_1^{'},\underline{Q}_{s_1},\cdots,\underline{Q}_{s_{k-1}}])$ is a non-zero polynomial. In what follows, we first show
that $\det([\underline{Q}_1^{'},\underline{Q}_{s_1},\cdots,\underline{Q}_{s_{k-1}}])$ satisfying
the subgraph of the optimization process is not identically zero and then discuss the required field size.

\begin{lem} \label{lem_nonzero2} For regenerating node $1$, there exist linear codes satisfying the minimum-cost subgraph (resulted from problem (\ref{opt-lin})) such that the polynomial $\det([\underline{Q}_1^{'},\underline{Q}_{s_1},\cdots,\underline{Q}_{s_{k-1}}])$ is non-zero for any selected set $\Xi_{k-1}$.
That is,
\begin{equation}
\prod_{\{s_1,\cdots,s_{k-1}\}\subset \{2,\cdots,n\}}
\det([\underline{Q}_1^{'},\underline{Q}_{s_1},\cdots,\underline{Q}_{s_{k-1}}])\neq 0.
\label{Pi}\end{equation} \end{lem}

\begin{proof}  Consider $\Xi_k=\{\underline{Q}_{s_1},\cdots,\underline{Q}_{s_{k}}\}$ a set of coding coefficients selected from $k$ out of $n$ nodes. Since every $k$ nodes can reconstruct the original file, then the matrix $[\underline{Q}_{s_1},\cdots,\underline{Q}_{s_{k}}]$ has full rank $M=k\alpha$. Thus, for $\Xi_{k-1}$, the matrix  $[\underline{Q}_{s_1},\cdots,\underline{Q}_{s_{k-1}}]$ has rank $(k-1)\alpha$.  To regenerate a new node with the RCP, the minimum-cost subgraph of the information flow graph  meets the requirement that by connecting the $DC$ to the new node and any $k-1$ other nodes, all the cuts are greater
than (or equivalent to) $M$.  This requires  the matrix of coding coefficients of the cut in the modified flow graph containing the new node and selection node set having full rank $M$. To prove this, consider a set $\mathcal{V}$ containing the data collector,  $in$ and $out$ nodes of the new node, and $out$ nodes of the nodes in set $\Xi_{k-1}$. Other nodes including the source node are in the complement set $\overline{\mathcal{V}}$.   The cut  passes the $in$-$out$ edges of nodes in $\Xi_{k-1}$ has the rank $(k-1)\alpha$. Since  all the cuts has at least $M$ edges, there would be at least $R=M-(k-1)\alpha=\alpha$ edges from  $\overline{\mathcal{V}}$ entering $\mathcal{V}$. In $\overline{\mathcal{V}}$, there exist $\alpha$  vectors e.g., vectors in  $\underline{Q}_{s_{k}}$ which are independent of vectors in $\Xi_{k-1}$.  Thus, if we send fragments corresponding to those $\alpha$ independent vectors through $R=\alpha$ edges to the $\mathcal{V}$, the matrix of the coding coefficients of the cut will be full rank. Therefore, $\det([\underline{Q}_1^{'},\underline{Q}_{s_1},\cdots,\underline{Q}_{s_{k-1}}])$ can be non-zero.  \end{proof}

%Assuming the full rank-cut as a virtual source ($S^{'}$) of original file, there would be a network on which $S^{'}$ is connected to  $(k-1)$ selected nodes with the links of finite-capacity, and all the cuts are greater than or equal to $M$. Therefore, new node beside those selected node can have rank $M$ or in other words, $\det([\underline{Q}_1^{'},\underline{Q}_{s_1},\cdots,\underline{Q}_{s_{k-1}}])\neq 0$.

 %Finally, by selecting code coefficients from proper finite field , which further is given, in a induction approach (as discussed by Koetter and Medard in \cite{KoMed}), the new node is regenerated satisfying RCP.\end{proof}

To find the required field size of OCMSR codes, we need to know the maximum degree of the variables of the polynomial in ($\ref{Pi}$). When SNC is allowed in intermediate nodes, as in our repairing scheme, the polynomial degree can be greater than non-SNC schemes since SNC may involve extra network coding processes, compared to the repairing schemes only encoding in the storage nodes and the new node.  For analysis, we use $n_{nc}$ to denote the maximum number of network coding processes involved by one fragment in one repair stage. We note that $n_{nc} \leq |N|$, the number of nodes in the networks.

\textit{Example 1:}  In the repair process of node four in tandem
network (Fig. \ref{modifiedgraph}), $n_{nc} =4$ . The  network coding is at node $1$, $2$, $3$ and at the new node.

\begin{thm} For a distributed storage system
$G(n,k,\alpha)$ with the source file of size $M$, if the
field size is greater than $d_0$, there exists a linear network code
such that the RCP is satisfied for any stage of repair, regardless
of how many failures/repairs happened before, where
\begin{equation}
d_0=\binom{n}{k}Mn_{nc}.
\end{equation}\label{Theorem:Fieldsize}
\end{thm}
\begin{proof} The proof is by induction on the number of repair stages. That is,  we assume before a node fails all the storage nodes have the RCP. In each stage of repair, when a node fails, the new node is regenerated preserving the RCP. Thus, we initialize the code on $n$ nodes by which any $k$ out of
$n$ nodes can reconstruct the original file. Then if
a node fails, the new node is regenerated such that  the repairing cost is minimized and the RCP is preserved.
By the RCP, the coding coefficients of any $k$ nodes must have full rank $M$. That is,
\begin{equation}
\prod_{\{s_1,\cdots,s_k\}\subset \{1,\cdots,n\}} \det ( [
\underline{Q}_{s_1}, \cdots, \underline{Q}_{s_k} ] ) \neq 0.
\label{Eqn:Repair-Matrix}
\end{equation}

Clearly, the maximum degree of variables in (\ref{Eqn:Repair-Matrix}) is $\binom{n}{k}M$. Thus, by Lemma 1, if the field size ($q$) is greater than $\binom{n}{k}M$, then there
is a network coding solution for repair.  Since $n_{nc}\geq 2$ (at least two coding process: one in surviving nodes, and another in the new node),
$d_0 \geq \binom{n}{k}M$; thus, there is a coding solution for $q \geq d_0$.

When a node fails (assume $Q_1$), the optimization algorithm finds
the minimum-cost subgraph. Accordingly, the fragments are combined using linear network coding, and then
 the new node is regenerated. The set including the new node ($Q_1^{'}$) and surviving
nodes must satisfy the RCP. Thus,
\begin{equation}
\prod_{\{s_1,\cdots,s_{k-1}\}\subset \{2,\cdots,n\}}
\det([\underline{Q}_1^{'},\underline{Q}_{s_1},\cdots,\underline{Q}_{s_{k-1}}])\neq 0.
\end{equation}\label{eqTH}
By Lemma 2, the polynomial can be nonzero. The maximum degree of each variable is less than $\binom{n-1}{k-1}M n_{nc}$. By Lemma 1, if the finite field size $q
\geq \binom{n-1}{k-1}M n_{nc}$, there is a network solution for the repair.
Clearly, $d_0=\binom{n}{k}Mn_{nc} \geq \binom{n-1}{k-1}M n_{nc}$ for $n\geq k$. Hence, for $q>d_0$, there exist a code for the repair. This concludes our proof.\end{proof}

In summary, optimal-cost repair for the minimum storage regenerating  (OCMSR) code is given in two steps. First, the optimal-cost subgraph is found decoupled from coding by solving problem (\ref{opt-lin}). Then, we can construct the code of the new node by  random linear coding \cite{Ho01} or deterministic \cite{Jagi}  from the finite field size determined by Theorem \ref{Theorem:Fieldsize}.

\section{  Repair-Cost Lower Bound, and the gain of optimal-cost repair in  networks with given topologies}\label{sec:lowerbound}
In this section, we study the lower bound of the repair-cost
for distributed storage systems on  networks with given topologies. We aim at showing how considering networks and cost together result in a lower cost in the repair. We first apply our method to tandem, star and grid networks, where there might not exist direct links between new nodes to all the surviving nodes. Next, for the scenario that the new node has direct links to surviving nodes (as \cite{Dimk01}, \cite{Wu01}), we show that surviving node cooperation can reduce the cost. For the purpose, we study the repair cost in a fully connected network.  It may worth to note that our approach can be applied on any network, and for any cost on the links of the networks. The networks considered in this section are just examples to present the gain of optimal-cost repair. For simplicity, we assume in this section that links have unit cost of transmission, unless otherwise stated. We define the gain in our approach as the ratio of repair-cost in minimum-bandwidth approach (denoted as $\sigma_{non-opt}$) to the optimal-cost repair (denoted as $\sigma_{opt}$),
\begin{equation}
g_c=\frac{\sigma_{non-opt}}{\sigma_{opt}}.
\label{tandem_eq}
\end{equation}
    %of all connected network the optimal-cost repair Since we are interested to compare the results with the case that cost has not been considered in the repair (non-optimized-cost repair denoted as $\sigma_{non-opt}$), we assume the transmission cost of all link are unit cost. This gives us the gain of surviving node cooperation in the repair process. We denote $g_c$ as the surviving node cooperation gain and define as follows.
\subsection{Tandem Networks}\label{sec:ResTandem}
\subsubsection{Repair-Cost Lower Bound in Tandem Networks}\label{sec:ResTandem}
In a tandem topology, nodes are in a line.  That is, each
node is linked to two neighboring nodes, except the nodes in the line ends, which have only one neighbor. When a node fails and a new
node joins, the repair traffic is relayed by intermediate nodes
to the new node. Then we can formulate the repair-cost lower bound by the following proposition.

\begin{pro} Consider a tandem distributed storage network consisting of
$n$ nodes where each node stores $\alpha$ fragments and every $k$
nodes can reconstruct the original file of size $M$. Assuming the cost of
a link (between adjacent nodes) equals to one unit, the
lowest repair-cost is achieved by cooperation of the $k$ nearest
surviving nodes and equals to
\begin{equation}
\sigma_c \geq  [k(M-(k-1)\alpha)]^{+},
\label{tandem_eq}
\end{equation}
where $[x]^{+}=\max\{x,0\}$
\label{pro:TandemCost}\end{pro}
\proof{ See appendix A.}

By Proposition \ref{pro:TandemCost}, we can calculate the gain of optimized cost  to the non-optimized cost approach. Gain in the repair for one of the end nodes in line can be calculated as follows.
\begin{col} Consider a distributed storage system with parameters $(M=k(n-k),\alpha=(n-k), d=n-1)$  in a tandem network with $n$ nodes in order (node $1$-node $2$-$\cdots$-node $n$). If  node $1$  fails, the surviving node cooperation gives the gain
\begin{equation} g_c=\frac{n(n+1)}{2k(n-k)}.\end{equation}
\label{pro:TandemSNCgain}\end{col}
\proof{ With the given parameters the minimum bandwidth repair is $\beta=1$ \cite{Dimk01}. Sending the repair bandwidth from node $n$ to the new node costs $n$ units, from node $n-1$ to the new node costs $n-1$ units and so on. Thus in total the repair cost in this approach is $\sigma_{non-opt}=n+(n-1)+\cdots+1=\frac{n(n+1)}{2}$. From Proposition \ref{pro:TandemCost}, $\sigma_c=k(n-k)$. Hence the gain would be $g_c=\frac{n(n+1)}{2k(n-k)}$.}

Therefore, for $(n,k)-$MDS codes with $k=n/2$ the $g_c$ tends to $2$ for a large network ($n\rightarrow \infty$).

\subsubsection{Explicit Construction of OCMSR code for Tandem Networks}
Previously, we proved that the optimal cost  functional repair can be achieved by applying linear network coding in a large finite field. In this section, we show that the lower bound can be achieved even with an extra condition of exact repair. We give an explicit construction for the optimal-cost exact repair in tandem networks. The motivation is that it requires  smaller finite field size comparing to random code and also it has explicit construction. We split the source file of size $M$ into $k$ fragments. We denote the source file by vector $\underline{m}=[m_1 m_2 \cdots m_k]^T$. To construct $(n,k)-MDS$ code, We use a $k \times n$ Vandermonde matrix  $G$ as a generator matrix.
\begin{equation}
G=  \left(\begin{array}{cccc}
1 & 1 & \cdots & 1\\
\alpha_1 & \alpha_2 &   \cdots & \alpha_n \\
\alpha_1^2 & \alpha_2^2 & \cdots & \alpha_n^2 \\
\vdots & \vdots & \ddots & \vdots\\
\alpha_1^{k-1} & \alpha_2^{k-1} &  \cdots & \alpha_n^{k-1}
\end{array}\right),\end{equation}
where  $\alpha_i$s for $i \in \{ 1,\cdots, n\}$ are distinct elements from the finite field $GF(q)$.

By the property of Vandermonde matrix, every $k \times k$ submatrix of $G$ is full rank if $\alpha_i$s, for $ i \in \{ 1,\cdots, n\}$, are distinct elements. This requires $q \geq n$. Each column i.e, column $t, t \in {1,\cdots,n}$ in matrix $G$ represents the code on node $t$. We denote the coded data on node $t$ as $v(t)$, then,
 \begin{equation}
 v_{t}=m_1+m_2 \alpha_{t}+\cdots+m_k \alpha_{t}^{k-1}= [1\text{ } \alpha_t \text{ } \cdots \alpha_t^{k-1}]\underline{m}.
\end{equation}

 By the property of Vandermonde matrix, every $k \times k$ submatrix of $G$ is full rank then a data collector can reconstruct the source file by connecting to any $k$ nodes.  The regenerating process is only by linear combination. Assume nodes are labeled in order i.e., node $1$ connects to node $2$, node $2$ connects to node $1$ and $3$, and so on. By Proposition  \ref{pro:TandemCost}, for $M=k$ fragments, the optimum-cost repair is by transmitting $M/k=k/k=1$ fragment to the neighbor. Assume node $t$, for $t \in \{1,\cdots,n\}$, fails and a sequence of backward nodes $\{ \texttt{node}_{t-k_{1}},\texttt{node}_{t-k_{1}+1},\cdots, \texttt{node}_{t-1} \}$ and a sequence of forward nodes $\{ \texttt{node}_{t+1},\texttt{node}_{t+2},\cdots, \texttt{node}_{t+k_{2}} \}$, such that $k_{1}+k_{2}=k$, help the new node to regenerate the content of the failed node.
%To regenerate the new node having exact code and data as the failed node.

The repair process is as follows. The new node receives fragments from two directions, from node $t-1$, aggregating data from nodes $\texttt{node}_{t-k_{1}}$ till \texttt{node} $t-1$,  and from \texttt{node} $t+1$, aggregating data from nodes $\texttt{node}_{t+k_{2}}$ to \texttt{node} $t+1$. Thus, in one direction node $\texttt{node}_{t-k_{1}}$ multiplies its content by a coefficient $\xi_{t-k_{1}}$ from $GF(q)$ and sends the result to node $\texttt{node}_{t-k_{1}+1}$. Then node $\texttt{node}_{t-k_{1}+1}$, multiplies its content by $\xi_{t-k_{1}+1}$ and  combines the result to the received fragment and then sends its combined fragment to its next neighbor $\texttt{node}_{t-k_{1}+2}$.  Finally  \texttt{node} $t-1$  transmits the combined fragment $w_{t-1}$, which is,
\begin{equation}\begin{array}{ll}
w_{t-1}=\xi_{t-k_{1}} v_{t-k_{1}}+\xi_{t-k_{1}+1} v_{t-k_{1}+1}+ \cdots + \xi_{t-1} v_{t-1}\\
=\xi_{t-k_{1}}(m_1+m_2 \alpha_{t-k_{1}}+\cdots+m_k \alpha_{t-k_{1}}^{k-1})\\
+\xi_{t-k_{1}+1}(m_1+m_2 \alpha_{t-k_{1}+1}+\cdots+m_k \alpha_{t-k_{1}+1}^{k-1})\\
+\cdots+\xi_{t-1}(m_1+m_2 \alpha_{t-1}+\cdots+m_k \alpha_{t-1}^{k-1}).
\end{array}\end{equation}
In another route, node \texttt{node} $t+1$ sends the aggregated fragment to the new node. That is,  $\texttt{node}_{t+k_{2}}$ multiplies its content by a coefficient $\xi_{t+k_{2}}$ from $GF(q)$ and sends the result to  $\texttt{node}_{t+k_{2}-1}$. Then node $\texttt{node}_{t+k_{2}-1}$, multiplies its content by $\xi_{t+k_{1}-1}$ and  combines the result to the received fragment and then sends its combined fragment to its next neighbor $\texttt{node}_{t+k_{2}-2}$.  Other helping nodes do the same. Finally  \texttt{node} $t+1$  transmits the combined fragment $w_{t+1}$, which is,
\begin{equation}\begin{array}{ll}
w_{t+1}=\xi_{t+k_{2}} v_{t+k_{2}}+\xi_{t+k_{2}-1} v_{t+k_{2}-2}+ \cdots + \xi_{t+1} v_{t+1}\\
=\xi_{t+k_{2}}(m_1+m_2 \alpha_{t+k_{2}}+\cdots+m_k \alpha_{t+k_{2}}^{k-1})\\
+\xi_{t+k_{2}-1}(m_1+m_2 \alpha_{t+k_{2}-1}+\cdots+m_k \alpha_{t+k_{2}-1}^{k-1})\\
+\cdots+\xi_{t+1}(m_1+m_2 \alpha_{t+1}+\cdots+m_k \alpha_{t+1}^{k-1}).
\end{array}\end{equation}
In order to achieve exact repair, we put $w_{t-1}+w_{t+1}=v_t$
\begin{equation}
 v_{t}=m_1+m_2 \alpha_{t}+\cdots+m_k \alpha_{t}^{k-1}.
\end{equation}
Thus, vector $\underline{\xi}=[\xi_{t-k_{1}},  \cdots, \xi_{t-1}, \xi_{t+1},  \cdots, \xi_{t+k_{2}}]^{T}$ should be selected such that,\\

$[\xi_{t-k_{1}} \xi_{t-k_{1}+1} \cdots, \xi_{t+k_{2}}] \times $
\begin{equation}
  \underbrace{\left(\begin{array}{ccccc}
1 & \alpha_{t-k_{1}} & \alpha_{t-k_{1}}^2 & \cdots & \alpha_{t-k_{1}}^{k-1}\\
1 & \alpha_{t-k_{1}+1} & \alpha_{t-k_{1}+1}^2 & \cdots & \alpha_{t-k_{1}+1}^{k-1} \\
\vdots & \vdots & \vdots & \ddots & \vdots \\
1 & \alpha_{t+k_{2}} & \alpha_{t+k_{2}}^2 & \cdots & \alpha_{t+k_{2}^{k-1}}\\
\end{array} \right)}_{\underline{A}}=\\
 \left(\begin{array}{c}
 1\\
 \alpha_{t}\\
  \alpha_{t}^2\\
   \vdots\\
     \alpha_{t}^{k-1}\end{array} \right)^{T}.\end{equation}
Since  matrix $\underline{A}$ is non-singular, we can always find linear codes ($\underline{\xi}$) that make the exact repair possible. Hence, for the successful reconstruction and repair process the finite field  $q > n$ suffices.

\subsection{Star Networks}\label{sec:RstStar}
\begin{figure}%[b]
 \centering
 \psfrag{a}[][][1.5]{ $\alpha$ }
 \psfrag{m1}[][][2]{ node 1 }
 \psfrag{m2}[][][2]{ node 2 }
 \psfrag{m3}[][][2]{ node 3 }
 \psfrag{m4}[][][2]{ node 4 }
 \psfrag{m5}[][][2]{ node 5 }
 \psfrag{m6}[][][2]{ node 6 }
 \psfrag{x1}[][][2]{ $x_1$ }
 \psfrag{x2}[][][2]{ $x_2$ }
 \psfrag{x3}[][][2]{ $x_3$ }
 \psfrag{x4}[][][2]{ $x_4$ }
 \psfrag{x5}[][][2]{ $x_5$ }
 \psfrag{x6}[][][2]{ $x_6$ }
 \psfrag{x7}[][][2]{ $x_7$ }
 \psfrag{x8}[][][2]{ $x_8$ }
 \psfrag{y1}[][][2]{ $y_1$ }
 \psfrag{y2}[][][2]{ $y_2$ }
 \psfrag{y3}[][][2]{ $y_3$ }
 \psfrag{y4}[][][2]{ $y_4$ }
 \psfrag{p1}[][][2]{ $p_1$ }
 \psfrag{p2}[][][2]{ $p_2$ }
 \psfrag{p3}[][][2]{ $p_3$ }
 \psfrag{p4}[][][2]{ $p_4$ }
 \psfrag{p5}[][][2]{ $p_5$ }
 \psfrag{p6}[][][2]{ $p_6$ }
 \psfrag{p7}[][][2]{ $p_7$ }
 \resizebox{6cm}{!}{\epsfbox{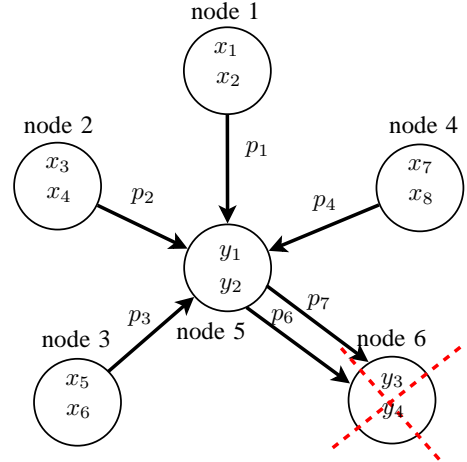}}
\caption{Optimal-cost repair in a star network with $6$ nodes.}
 \label{star}
\end{figure}
In a star network, there is a central node, and any pair of non-central nodes must
communicate through the central node.
 Fig. \ref{star} shows an example of a distributed
 storage system in a star network. In the scenario, if the central node fails, and all link costs are equivalent, the optimal-cost repair reduces to the optimal-bandwidth repair in \cite{Dimk01}. However, if non-central nodes fail, the situation is different. Assuming $n$ storage nodes and a source file size $M$, the lower bound of  repair-cost  for a non-central node is calculated as follows.

\begin{pro} In a  distributed storage system with a star topology,
the repair-cost of a non-central is greater than or equals to
\begin{equation}
\sigma_c \geq (\frac{n-2}{n-k}+1) [M-(k-1)\alpha]^{+}.
\end{equation}
\label{star-pro}
 \label{pro:starCost}\end{pro}
\proof{ See appendix B.}

By Proposition \ref{pro:starCost}, we can calculate the gain of surviving node cooperation comparing to the non-optimized cost approach for non-central nodes, as follows.
\begin{col} Consider a distributed storage system with parameters $(M=k(n-k),\alpha=(n-k), d=n-1)$ in a star network with $n$ nodes. If a non-central node fails, the surviving node cooperation gives the gain $g_c$,
\begin{equation} g_c=\frac{2n-3}{2n-k-2}.\end{equation}
\label{pro:starSNCgain}\end{col}
\proof{With the given parameters the minimum bandwidth repair is $\beta=1$ \cite{Dimk01}. Sending the repair traffic from $(n-2)$ non-central nodes to the central node costs $(n-2)$. The central node sends these $n-2$ fragments plus its own repair. This means the central node sends $(n-1)$ fragments to the new node. In total the repair cost equals to $\sigma_{non-opt}=(n-2)+(n-1)=2n-3$. For the optimal cost, we calculate from Proposition \ref{pro:starCost}, $\sigma_c = (\frac{n-2}{n-k}+1)\frac{M}{k}$.} Substituting $M=k(n-k)$ yields $\sigma_c =2n-k-2$.

Therefore, as an example for $k=n/2$, $g_c$ tends to $\frac{4}{3}$ for $n\rightarrow \infty$.

\subsection{Grid Networks}\label{sec:rstGrid}

Consider a $2\times 3$ grid network in Fig.  \ref{grid-NC}. The optimal-cost repair process
 depends on the location of the new node. The optimum cost repair for the repair in node $6$ can be found by our approach.

\begin{clm} In the repair of node $6$ in Fig. \ref{grid-NC}, the optimal-cost repair is $7$ units corresponding to the minimum-cost subgraph
subgraph $(z_{(12)}, z_{(14)}, z_{(23)}, z_{(25)}, z_{(36)},
z_{(45)}, z_{(56)})=(0, 1, 0, 1, 1, 2, 2)$.\end{clm}

\proof{ See appendix C.}

In this example in the non-optimized cost approach, the minimum bandwidth for $M=8, n=6, k=4, d=n-1$ is $\beta=1$ fragment \cite{Dimk01}. Surviving nodes transmit their fragment to the new node following this approach: node $3$ by the path $\texttt{node} 1 \rightarrow \texttt{node} 2 \rightarrow \texttt{node} 3 \rightarrow 6$, node $2$ by $\texttt{node} 2 \rightarrow \texttt{ node} 3 \rightarrow 6$, node $3$ by $\texttt{node} 3 \rightarrow 6$, node $4$ by $\texttt{node} 4 \rightarrow  \texttt{node} 5 \rightarrow 6$, and finally node $5$ by transmitting a fragment on link $\texttt{node} 5 \rightarrow 6$. This gives $\sigma_{non-opt}=3+2+1+2+1=9$ units. Thus,  our approach gives a gain $g_c=\frac{9}{7}=1.2$.

Finding the lower bound of the repair
cost in a closed form for a grid network is more complicated than those for tandem
   and star networks. Yet, we can know  the repair-cost in
   a grid network will not be greater
   than the repair-cost in the  tandem topology if other conditions are the same, e.g., link costs, the number of nodes. More formally, we have
\begin{col} Optimal-cost repair for a distributed storage system in a grid
network with $n$ node coded by $(n,k)-$ MDS codes (each node stores $M/k$ fragments) leads less repair-cost than that of a tandem topology if all other conditions are the same. That is,
\begin{equation}
{\sigma_c}^{grid} \leq {\sigma_c}^{tandem}=[k(M-(k-1)\alpha)]^{+}.
\end{equation}
 \label{col:grid23}\end{col}
\begin{proof} The repair process  in a grid network can be reduced to that of  tandem networks if we neglect some available links which may reduce the cost. Thus, the cost of optimal-cost repair of a grid network is upper-bounded by that of a tandem topology. \end{proof}
%By Colory \ref{col:grid23}, we can calculate the minimum gain of surviving node cooperation comparing to the non-optimized cost approach, as follows.
%
%\begin{pro} Consider a tandem network with $n$ nodes, when a node fails, the surviving node cooperation gives the gain $g_c$,
%\end{pro}{\label{pro:gridSNCgain}}

\begin{figure}%[b]
 \centering
 \psfrag{a}[][][1.5]{ $\alpha$ }
 \psfrag{m1}[][][2]{ node 1 }
 \psfrag{m2}[][][2]{ node 2 }
 \psfrag{m3}[][][2]{ node 3 }
 \psfrag{m4}[][][2]{ node 4 }
 \psfrag{m5}[][][2]{ node 5 }
 \psfrag{m6}[][][2]{ node 6 }
 \psfrag{x1}[][][2]{ $x_1$ }
 \psfrag{x2}[][][2]{ $x_2$ }
 \psfrag{x3}[][][2]{ $x_3$ }
 \psfrag{x4}[][][2]{ $x_4$ }
 \psfrag{x5}[][][2]{ $x_5$ }
 \psfrag{x6}[][][2]{ $x_6$ }
 \psfrag{x7}[][][2]{ $x_7$ }
 \psfrag{x8}[][][2]{ $x_8$ }
 \psfrag{y1}[][][2]{ $y_1$ }
 \psfrag{y2}[][][2]{ $y_2$ }
 \psfrag{y3}[][][2]{ $y_3$ }
 \psfrag{y4}[][][2]{ $y_4$ }
 \psfrag{p1}[][][2]{ $p_1$ }
 \psfrag{p2}[][][2]{ $p_2$ }
 \psfrag{p3}[][][2]{ $p_3$ }
 \psfrag{p4}[][][2]{ $p_4$ }
 \psfrag{p5}[][][2]{ $p_5$ }
 \psfrag{p6}[][][2]{ $p_6$ }
 \psfrag{p7}[][][2]{ $p_7$ }
 \resizebox{6cm}{!}{\epsfbox{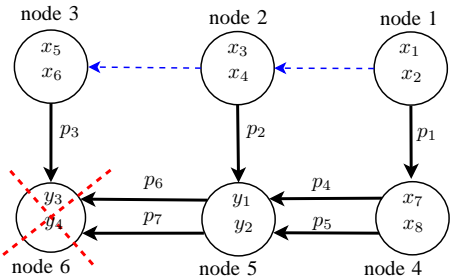}}
\caption{Optimization repair in the $2\times 3$  grid network.
Dashed lines show available links which are not used in the repair
process.}
 \label{grid-NC}
\end{figure}

\subsection{Fully connected network}\label{sec:ResTandem}
In a fully connected network all nodes have direct links to each other. This model suits for data centers where each  data center is connected to another one through a hierarchical network structure \cite{Benson}. For an example, consider a network with $5$ nodes where every pair of nodes are connected by direct links. We shall show that there is a gain in SNC also in fully connected networks. For this, first assume the transmission of one unit of data on all links costs one unit. In this case, the optimal cost approach is to transmit from surviving nodes to the new node (without surviving node cooperation) and the cost  is equivalent to the approach in  \cite{Dimk01}.
\begin{pro} For a distributed storage in Fig.\ref{Fig:FullyConnectedNet}, with parameters $(M=6,n=5,k=3,d=4,\alpha=M/k)$ in a fully connected network with equal transmission cost on links, the optimum-cost approach is  one-hop transmission, namely, directly from surviving nodes to the new node.\label{Pro-FC}\end{pro}
\proof{ See appendix D.}

 Proposition \ref{Pro-FC} holds for a fully connected network with all the links having equivalent costs. In the case that there are different costs, then it might not be optimum-cost to follow the same approach (direct  transmission). In this case it might be better to exploit the multi-hop network structure. For example, consider a scenario where links from surviving nodes to the new node have high cost, however links between surviving nodes are low in cost (e.g., some nodes which are located close together in one region want to send data to the new node located very far from them). In this scenario surviving nodes can cooperate and only send the aggregated data of size $\alpha$ to the new node. More formally, we have the following result.

\begin{pro} Consider a distributed storage in Fig.\ref{Fig:FullyConnectedNet} having same parameters as Proposition \ref{Pro-FC}. The only difference is that transmitting one fragment from surviving nodes to the new node costs $3$ units, i.e., $c_{(15)}=c_{(25)}=c_{(35)}=c_{(45)}=3$. Then the optimal repair cost equals to $10$ and the optimal-cost repair is through links $z_{(23)}=z_{(34)}=z_{(45)}=2$.\end{pro}

\proof{ See appendix E.}

Thus, $\sigma_{opt}=2\times 1+2\times 1+2\times 3=10$. Yet, the repair by direct transmission costs $12$ units, which is suboptimal. Hence, the optimal-cost approach, using surviving node cooperation, gains $\frac{12}{10}$ in reducing the repair-cost.\label{Pro-FC2}

\begin{figure}%[b]
 \centering
 \psfrag{m1}[][][2]{ node 1 }
 \psfrag{m2}[][][2]{ node 2 }
 \psfrag{m3}[][][2]{ node 3 }
 \psfrag{m4}[][][2]{ node 4 }
 \psfrag{m5}[][][2]{ node 5 }
 \resizebox{4cm}{!}{\epsfbox{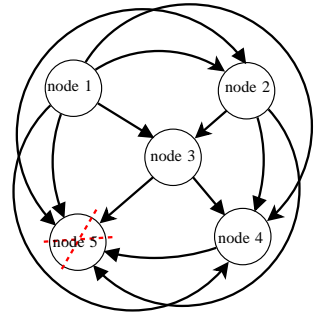}}
\caption{Optimal-cost repair in a fully connected network with $5$ nodes.}
 \label{Fig:FullyConnectedNet}
\end{figure}

\section{Conclusions}\label{sec:conclusion}

We have studied repair costs for distributed storage systems with network coding in heterogenous networks. We
formulated a linear programming problem which gives the fundamental lower bound of the repair cost. The linear programming also leads us to the optimal cost approach for the minimum storage per node, (OCMSR) codes. We prove the existence of the code by the random linear code in a large finite field. We
discussed the required field size for the existence of the code. To reduce the cost in networks, we proposed surviving node
cooperation. We also proposed an explicit construction for the exact optimal-cost repair in a tandem network. For specific networks with tandem, grid and star topology, we give the closed-form gains of optimal repair. We also discuss the gain of our approach for networks with fully connected toplogy.

\section{appendices}
\subsection{Proof of
Proposition 2}
Intuitively, every $k$ nodes can reconstruct the original file.
Thus, if $k$ nodes are allowed to combine their fragments
(possibly using SNC) then they can regenerate the new node.
When the cost of transmission on each channel is the same, the $k$
closest neighbors to the new node use the least transmission costs on the
network.  Thus, they may be the optimal-cost solution. A  stricter proof is as follows.

Assume $n$ nodes are labeled from $1$ to $n$ as shown in Fig. \ref{app1}. Without loss of generality, we assume node $1$ fails and the new node $(1^{'})$ is regenerated. We  denote the constraint region in a general form and then find the minimum repair cost value. We obtain the constraint region by the following steps:

Step $1$: connect the DC to the new node and nodes $n,\cdots,n-k+1$, as shown in Fig. \ref{app1}. Then the cut constraint is,
\begin{equation}
 z_{(21^{'})}+(k-1)\alpha \geq M \Rightarrow z_{(21^{'})} \geq M-(k-1)\alpha. \\
 \label{act1}
\end{equation}

Step $2$:  connect  the DC to the new node, node $2$ and  nodes $n,\cdots, n-k+2$. Then the cut constraint is,
\begin{equation}
 z_{(32)}+(k-1)\alpha \geq M \Rightarrow z_{(32)} \geq M-(k-1)\alpha. \\
 \label{act2}
\end{equation}

Step $3$: connect the DC to the new node, node $2, \cdots,j-1$  and  nodes $n,\cdots, n-k+j-1$ for $j=3, \cdots, k+1$. Then the cut constraint is,

\begin{equation}
 z_{(jj-1)}+(k-1)\alpha \geq M \Rightarrow z_{(jj-1)} \geq M-(k-1)\alpha,
 \label{act3}
\end{equation}

Other constraints of connecting DC to the new node and $k-1$ surviving nodes are non-active constraints \footnote{A non-active constraint is a constraint that does not affect or change the constraint region.} in the problem. Finally, the active constraints are
\begin{equation}
\begin{cases}
z_{(21^{'})} \geq M-(k-1)\alpha, \\
 z_{(32)} \geq M-(k-1)\alpha,\\
 \vdots\\

 z_{(k+1k)} \geq M-(k-1)\alpha, \\
  z_{(k+2k+1)} \geq 0,\\
  \vdots
\\
 z_{(n-1n-2)} \geq 0,\\
 z_{(nn-1)} \geq 0.\\

\end{cases}
\label{eq-active-c}
\end{equation}

Therefore, summing both sides of inequalities yields,
\begin{equation}
\sigma_c= z_{(21^{'})}+ z_{(32)}+\cdots+z_{(nn-1)}\geq k(M-(k-1)\alpha).
\end{equation}

\begin{figure}%[b]
 \centering
 \psfrag{a}[][][1.5]{ $\alpha$ }
 \psfrag{inf}[][][3]{ $\infty$ }
 \psfrag{s}[][][3]{ $S$ }
 \psfrag{DC}[][][3]{ $DC$ }
 \psfrag{vdvd}[][][3]{ $\vdots$ }
 \psfrag{a}[][][3]{ $\alpha$ }
  \psfrag{nn}[][][3]{ node $n$ }
   \psfrag{nn1}[][][3]{ node $n-1$ }
    \psfrag{nnk}[][][3]{ node $n-k+1$ }
     \psfrag{nk}[][][3]{ node $k$ }
      \psfrag{n2}[][][3]{ node $2$ }
       \psfrag{n1}[][][3]{ node $1$ }
        \psfrag{n1p}[][][3]{ node $1^{'}$ }
         \psfrag{z21p}[][][3]{  $z_{21^{'}}$ }
         \psfrag{cut}[][][3]{  Cut }
 \resizebox{8cm}{!}{\epsfbox{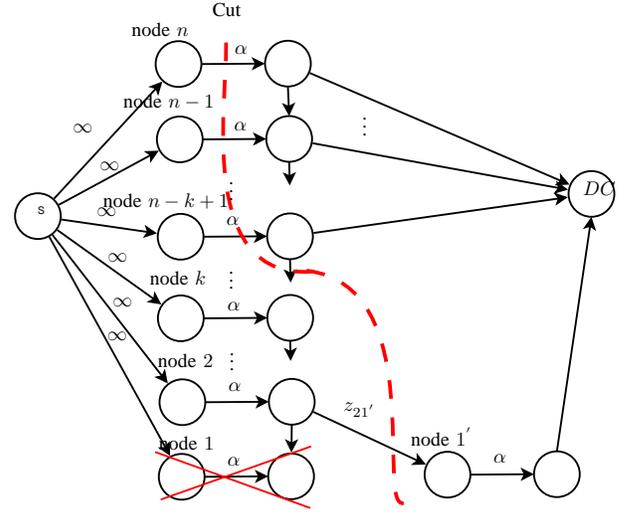}}
\caption{Cut in tandem network. This figure corresponds to the cut of $z_{21^{'}}+(k-1)\alpha\geq M$.}
  \label{app1}
\end{figure}
\subsection{Proof of
Proposition 3}
%For $k=2$, the proof is same as the tandem network, i.e., each node stores $\alpha \geq M/k=M/2$ fragments. Thus, the repair would be by cooperation of two nodes have access to the source file, . Thus, with $\sigma_c\geq 2(M-\alpha)$.

We draw the information flow graph, as shown in Fig. \ref{app2}, and prove the results by the following steps:
\begin{itemize}
\item Step 1: connect the DC to the new node, and nodes $n,\cdots,n-k+1$, as shown in Fig. \ref{app2}. Then the cut constraint is,
\begin{equation}
z_{(21^{'})}+(k-1)\alpha \geq M,\\
\Rightarrow z_{(21^{'})} \geq M-(k-1)\alpha.
\end{equation}

\item Step 2: connect the DC to the new node, and node $2$ and $(k-2)$ other nodes. We see for all $\binom{n-2}{k-2}$ selections of  other nodes, each $z_{(ij)}$ appears $\binom{n-3}{k-2}$. Thus if we add both sides of the constraints, we have
 \begin{equation}
\binom{n-3}{k-2}(z_{n2}+\cdots+z_{32})+\binom{n-2}{k-2}(k-1)\alpha \geq \binom{n-2}{k-2}M,
\end{equation}
 \begin{equation}
\Rightarrow \sigma_c-z_{(21^{'})} \geq \frac{\binom{n-2}{k-2}}{\binom{n-3}{k-2}}( M-(k-1))\alpha.
\end{equation}
\item Step 3: Any other cut constraint is non-active constraint.
\end{itemize}
therefore,
\begin{equation}
\sigma_c \geq (\frac{\binom{n-2}{k-2}}{\binom{n-3}{k-2}}+1) (M-(k-1)\alpha)=(\frac{n-2}{n-k}+1)(M-(k-1)\alpha).
\end{equation}

\begin{figure}%[b]
 \centering
 \psfrag{a}[][][1.5]{ $\alpha$ }
 \psfrag{inf}[][][3]{ $\infty$ }
 \psfrag{s}[][][3]{ $S$ }
 \psfrag{DC}[][][3]{ $DC$ }
 \psfrag{vdvd}[][][3]{ $\vdots$ }
 \psfrag{a}[][][3]{ $\alpha$ }
  \psfrag{nn}[][][3]{ node $n$ }
   \psfrag{nn1}[][][3]{ node $n-1$ }
    \psfrag{nnk}[][][3]{ node $n-k+1$ }
     \psfrag{nk}[][][3]{ node $k$ }
      \psfrag{n2}[][][3]{ node $2$ }
       \psfrag{n1}[][][3]{ node $1$ }
        \psfrag{n1p}[][][3]{ node $1^{'}$ }
         \psfrag{z21p}[][][3]{  $z_{21^{'}}$ }
          \psfrag{cut}[][][3]{  Cut }
 \resizebox{8cm}{!}{\epsfbox{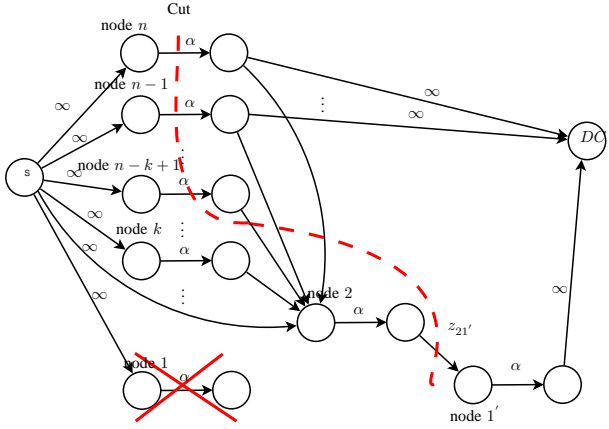}}
\caption{Cut analysis in star network. node $2$ is central node. This figures corresponds to the cut constraint of $z_{(21^{'})}+(k-1)\alpha \geq M$. }
  \label{app2}
\end{figure}

\subsection{Proof of
Claim 1}

The corresponding cost matrix ($\underline{C}$) for the repair on node $6$ is as below. We formulate the
linear optimization problem as follows. There are $\binom{6-1}{4-1}=10$ active cut constraints. Fig. \ref{Appx_C1} shows one of these cut constraints.

\begin{equation}
\underline{C}=
\begin{pmatrix}
0 &1 &\infty &1 &\infty &\infty \\
\infty &0 &1 &\infty &1 &\infty \\
\infty &\infty &0 &\infty &\infty &1 \\
\infty &\infty &\infty &0 &1 &\infty \\
\infty &\infty &\infty &\infty &0 &1
\end{pmatrix}.
\end{equation}

\begin{figure}[h!]
 \psfrag{a}[][][3]{ $\alpha$ }
 \psfrag{inf}[][][3]{ $\infty$ }
 \psfrag{s}[][][3]{ $S$ }
 \psfrag{DC}[][][3]{ $DC$ }
 \psfrag{a}[][][3]{ $\alpha$ }
 \psfrag{z12}[][][3]{  $z_{(12)}$ }
 \psfrag{z23}[][][3]{  $z_{(23)}$ }
 \psfrag{z14}[][][3]{  $z_{(14)}$ }
 \psfrag{z25}[][][3]{  $z_{(25)}$ }
 \psfrag{z36}[][][3]{  $z_{(36)}$ }
 \psfrag{z56}[][][3]{  $z_{(56)}$ }
      \psfrag{n2}[][][3]{ node $2$ }
       \psfrag{n1}[][][3]{ node $1$ }
        \psfrag{n2}[][][3]{ node $2$ }
         \psfrag{n3}[][][3]{ node $3$ }
          \psfrag{n4}[][][3]{ node $4$ }
           \psfrag{n5}[][][3]{ node $5$ }
            \psfrag{n6}[][][3]{ node $6$ }
             \psfrag{new6}[][][3]{ new node }
        \psfrag{n1p}[][][3]{ node $1^{'}$ }
         \psfrag{z21p}[][][3]{  $z_{21^{'}}$ }
          \psfrag{cut}[][][3]{  Cut }
     \resizebox{8cm}{!}{\epsfbox{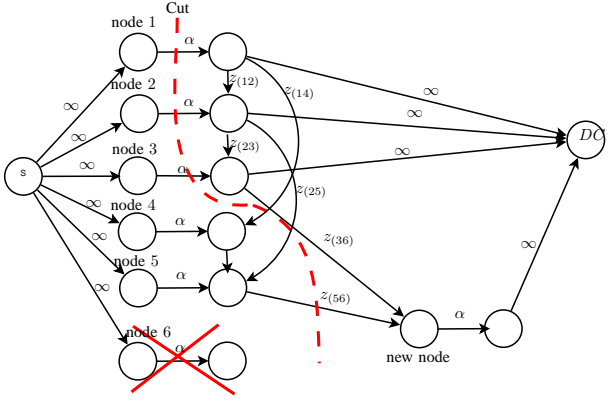}}
      \caption{Cut analysis in the $2\times3$ grid network with $n=6$, $k=4$. The cut in this figure corresponds to the inequality, $z_{(56)}\geq M-3\alpha$}
      \label{Appx_C1}
\end{figure}

\begin{equation}
 \begin{array}{ll}
\mbox{$\mathbf{\min}$} &\sigma_c(\underline{z}) \\
\mbox{$\mathbf{s.t.}$} & \underline{H} \text{ }.\underline{z}\geq (M-3\alpha)\underline{\mathbf{1}} \\
& z_{ij}\geq 0, \label{eqopt3}
\end{array}
\end{equation}
where $\underline{H}=$
\[
\left[\begin{array}{ccccccc}
0 & 0 & 0 & 0 & 1 & 0  & 1 \\
0 & 0 & 1 & 1 & 1 & 0  & 0 \\
0 & 0 & 1 & 0 & 0 & 0  & 1 \\
0 & 1 & 0 & 0 & 1 & 0  & 0 \\
1 & 1 & 0 & 0 & 0 & 0  & 1 \\
0 & 1 & 1 & 1 & 0 & 0  & 0 \\
0 & 0 & 0 & 0 & 1 & 1  & 0 \\
0 & 0 & 0 & 0 & 0 & 0  & 1 \\
0 & 0 & 1 & 1 & 0 & 1  & 0 \\
0 & 0 & 0 & 0 & 0 & 1  & 0
\end{array}\right]
\text{ and } \underline{z}=
 \begin{pmatrix}
z_{12}  \\
z_{14}\\
z_{23}\\
z_{25} \\
z_{36} \\
z_{45} \\
z_{56} \\

\end{pmatrix},\]
where $\underline{\mathbf{1}}$ is a $10 \times 1$ vector with all the entities equal to one). We note that matrix $H$ is resulted from cut-set analysis in the first stage of repair. For example the cut mentioned in Fig. \ref{Appx_C1} constructs the $8$-th row in matrix $H$, corresponding to inequality $z_{(56)}\geq M-3\alpha$.

Solving this linear optimization problem (e.g. by the simplex method) for the $M=8$ and $\alpha=2$ results:\\
$\sigma_c=7$;\\
$(z_{(12)}, z_{(14)}, z_{(23)},
z_{(25)}, z_{(36)}, z_{(45)}, z_{(56)})=\\(0, 1, 0, 1, 1, 2, 2).\blacksquare$

\subsection{Proof of
Proposition 4}
The cost matrix corresponding to regenerate the new node on node $5$ equals to,
\begin{equation}
\underline{C}=
\begin{pmatrix}
0 &1 &1 &1 &1\\
\infty &0 &1 &1 &1\\
\infty &\infty &0 &1 &1\\
\infty &\infty &\infty &0 &1
\end{pmatrix}.
\end{equation}
The cut analysis in first stage of repair give the following linear programming problem,
\begin{equation}
 \begin{array}{ll}
\mbox{$\mathbf{\min}$} &\sigma_c(\underline{z})=\mathbf{\underline{1}}^T \underline{z}\\
\mbox{$\mathbf{s. t.}$} & \underline{H} \text{ }.\underline{z}\geq (M-2\alpha) \mathbf{\underline{1}}, \\
& z_{(ij)}\geq 0, \label{eqopt1}
\end{array}
\end{equation}
where $\mathbf{\underline{1}}$ is a column vector with $10$ elements all equal to $1$, and T indicates vector transpose, and vector $\underline{z}=(z_{(12)},z_{(13)},z_{(14)},z_{(15)},z_{(23)},z_{(24)},z_{(25)},z_{(34)},z_{(35)},z_{(45)})$, and $\underline{H}$ calculated from cut set analysis equals to:
\begin{equation*}
\underline{H}=\left[\begin{array}{cccccccccc}
0 & 0 & 0 & 1 & 1 & 1  & 1 &0 & 0 & 0\\
1 & 0 & 0 & 1 & 0 & 0  & 0 &1 & 0 & 0\\
0 & 0 & 0 & 0 & 0 & 0  & 1 &1 & 1 & 0\\
1 & 0 & 0 & 1 & 0 & 0  & 0 &0 & 0 & 1\\
1 & 0 & 0 & 0 & 1 & 1  & 0 &0 & 0 & 0\\
0 & 0 & 0 & 0 & 0 & 0  & 0 &0 & 1 & 1\\
\end{array}\right].
\end{equation*}
Solving the linear programming problem for $M=6,\alpha=2$ gives the minimum repair cost as,
$z_{(15)}=z_{(25)}=z_{(35)}=z_{(45)}=1$ and zero for other links. This shows in this example the minimizing repair cost approach gives the same minimum value as minimizing repair bandwidth.

\subsection{Proof of
Proposition 5}
The cut-analysis is same as Proposition 4. The only difference is in the cost function. Thus, the optimization problem would be,
\begin{equation}
 \begin{array}{ll}
\mbox{$\mathbf{\min}$} &\sigma_c(\underline{z})=(1,1,1,3,1,1,3,1,3,3) \underline{z}\\
\mbox{$\mathbf{s. t.}$} & \underline{H} \text{ }.\underline{z}\geq (M-2\alpha) \mathbf{\underline{1}}, \\
& z_{(ij)}\geq 0, \label{eqopt22}
\end{array}
\end{equation}
where $\underline{H},M,\alpha$ are same as Proposition $4$. Solving the linear programming problem gives the optimum solution as  $z_{(23)}=z_{(34)}=z_{(45)}=2$ and $\sigma_c=(1,1,1,3,1,1,3,1,3,3) \underline{z}=z_{(23)}\times 1+z_{(34)}\times 1+ z_{(45)}\times 3 =10$.

\section*{Acknowledgement}
The authors would like to thank Dr. Majid Nasiri Khormoji for his helpful discussion and comments. The authors would also like to thank the anonymous reviewers
for their valuable comments and suggestions to improve the
quality of the paper.

\end{document}